\documentclass[11pt,aps,pra,notitlepage,tightenlines,nofootinbib,superscriptaddress]{revtex4-2}
\usepackage{orcidlink}
\usepackage{inputenc}

\usepackage{newpxtext,newpxmath}

\let\coloneqq\relax

\usepackage{amsthm}
\usepackage{amssymb}
\usepackage{amsmath}
\usepackage{bbold}
\usepackage{bbm}
\usepackage{braket}
\usepackage{dsfont}
\usepackage{mathdots}
\usepackage{mathtools}
\usepackage{enumerate}
\usepackage[shortlabels]{enumitem}
\usepackage{csquotes}
\usepackage{stmaryrd}
\usepackage[cal=boondox]{mathalfa}
\usepackage{graphicx}
\usepackage[capitalise]{cleveref}
\usepackage{stackengine}
\usepackage{scalerel}
\usepackage{tensor}       %\tensor[_n]{\braket{j|\psi}}{_{1\ldots n}}
\usepackage{array}
\usepackage{makecell}
\newcolumntype{x}[1]{>{\centering\arraybackslash}p{#1}}
\usepackage{tikz}
\usepackage{pgfplots}
\usetikzlibrary{shapes.geometric, shapes.misc, positioning, arrows, arrows.meta, decorations.pathreplacing, decorations.pathmorphing, patterns, angles, quotes, calc}
\usepackage{booktabs}
\usepackage{xfrac}
\usepackage{siunitx}
\usepackage{centernot}
\usepackage{comment}
\usepackage{chngcntr}
\usepackage{caption}
\usepackage{subcaption}

\newtheorem{thm}{Theorem}
\newtheorem*{thm*}{Theorem}

\newtheorem*{prop*}{Proposition}
\newtheorem{lemma}[thm]{Lemma}
\newtheorem*{lemma*}{Lemma}
\newtheorem{cor}[thm]{Corollary}
\newtheorem*{cor*}{Corollary}

\newtheorem*{cj*}{Conjecture}
\newtheorem{Def}[thm]{Definition}
\newtheorem*{Def*}{Definition}

\newtheorem*{question*}{Question}

\newtheorem*{problem*}{Problem}

\makeatletter
\def\thmhead@plain#1#2#3{%
  \thmname{#1}\thmnumber{\@ifnotempty{#1}{ }\@upn{#2}}%
  \thmnote{ {\the\thm@notefont#3}}}
\let\thmhead\thmhead@plain
\makeatother

\theoremstyle{definition}

\newcommand{\bb}{\begin{equation}\begin{aligned}\hspace{0pt}}
\newcommand{\bbb}{\begin{equation*}\begin{aligned}}
\newcommand{\ee}{\end{aligned}\end{equation}}
\newcommand{\eee}{\end{aligned}\end{equation*}}
\newcommand*{\coloneqq}{\mathrel{\vcenter{\baselineskip0.5ex \lineskiplimit0pt \hbox{\scriptsize.}\hbox{\scriptsize.}}} =}

\newcommand{\eqt}[1]{\stackrel{\mathclap{\scriptsize \mbox{#1}}}{=}}

\newcommand{\geqt}[1]{\stackrel{\mathclap{\scriptsize \mbox{#1}}}{\geq}}

\renewcommand{\epsilon}{\varepsilon}

\newcommand{\R}{\mathds{R}}

\DeclareMathOperator{\Tr}{Tr}

\DeclareMathAlphabet{\pazocal}{OMS}{zplm}{m}{n}

\DeclareMathOperator{\diag}{diag}

\newcommand{\lsmatrix}{\left(\begin{smallmatrix}}
\newcommand{\rsmatrix}{\end{smallmatrix}\right)}

\stackMath

\stackMath

\makeatletter
\newcommand*\rel@kern[1]{\kern#1\dimexpr\macc@kerna}
\newcommand*\widebar[1]{%
  \begingroup
  \def\mathaccent##1##2{%
    \rel@kern{0.8}%
    \overline{\rel@kern{-0.8}\macc@nucleus\rel@kern{0.2}}%
    \rel@kern{-0.2}%
  }%
  \macc@depth\@ne
  \let\math@bgroup\@empty \let\math@egroup\macc@set@skewchar
  \mathsurround\z@ \frozen@everymath{\mathgroup\macc@group\relax}%
  \macc@set@skewchar\relax
  \let\mathaccentV\macc@nested@a
  \macc@nested@a\relax111{#1}%
  \endgroup
}

\counterwithin*{equation}{part}
\counterwithin*{thm}{part}
\counterwithin*{figure}{part}

\tikzset{meter/.append style={draw, inner sep=10, rectangle, font=\vphantom{A}, minimum width=30, line width=.8, path picture={\draw[black] ([shift={(.1,.3)}]path picture bounding box.south west) to[bend left=50] ([shift={(-.1,.3)}]path picture bounding box.south east);\draw[black,-latex] ([shift={(0,.1)}]path picture bounding box.south) -- ([shift={(.3,-.1)}]path picture bounding box.north);}}}
\tikzset{roundnode/.append style={circle, draw=black, fill=gray!20, thick, minimum size=10mm}}
\tikzset{squarenode/.style={rectangle, draw=black, fill=none, thick, minimum size=10mm}}

\definecolor{Blues5seq1}{RGB}{239,243,255}
\definecolor{Blues5seq2}{RGB}{189,215,231}
\definecolor{Blues5seq3}{RGB}{107,174,214}
\definecolor{Blues5seq4}{RGB}{49,130,189}
\definecolor{Blues5seq5}{RGB}{8,81,156}

\definecolor{Greens5seq1}{RGB}{237,248,233}
\definecolor{Greens5seq2}{RGB}{186,228,179}
\definecolor{Greens5seq3}{RGB}{116,196,118}
\definecolor{Greens5seq4}{RGB}{49,163,84}
\definecolor{Greens5seq5}{RGB}{0,109,44}

\definecolor{Reds5seq1}{RGB}{254,229,217}
\definecolor{Reds5seq2}{RGB}{252,174,145}
\definecolor{Reds5seq3}{RGB}{251,106,74}
\definecolor{Reds5seq4}{RGB}{222,45,38}
\definecolor{Reds5seq5}{RGB}{165,15,21}

\allowdisplaybreaks

\begin{document}

\title{Extracting energy via bosonic Gaussian operations}

\author{Frank Ernesto Quintela Rodriguez}%\,\orcidlink{0000-0002-9475-2267}}
\email[Email: ]{frank.quintelarodriguez@sns.it}
\affiliation{Scuola Normale Superiore, Piazza dei Cavalieri, 7, Pisa, 56126, Italy}
\author{Francesco Anna Mele}%\,\orcidlink{0000-0002-7900-5167}}
\email[Email: ]{francesco.mele@sns.it}
\affiliation{Scuola Normale Superiore, Piazza dei Cavalieri, 7, Pisa, 56126, Italy}
\author{Salvatore Francesco Emanuele Oliviero}%\,\orcidlink{0000-0003-3569-085X}}
\email[Email: ]{salvatore.oliviero@sns.it}
\affiliation{Scuola Normale Superiore, Piazza dei Cavalieri, 7, Pisa, 56126, Italy}
\author{Vittorio Giovannetti}%\,\orcidlink{0000-0002-7636-9002}}
\email[Email: ]{vittorio.giovannetti@sns.it}
\affiliation{Scuola Normale Superiore, Piazza dei Cavalieri, 7, Pisa, 56126, Italy}
\author{Ludovico Lami}%\,\orcidlink{0000-0003-3290-3557}}
\email[Email: ]{ludovico.lami@sns.it}
\affiliation{Scuola Normale Superiore, Piazza dei Cavalieri, 7, Pisa, 56126, Italy}
\author{Vasco Cavina}%\,\orcidlink{0000-0001-9468-7410}}
\email[Email: ]{vasco.cavina@sns.it}
\affiliation{Scuola Normale Superiore, Piazza dei Cavalieri, 7, Pisa, 56126, Italy}

\begin{abstract}
Quantum thermodynamics is often formulated as a theory with constrained access to operations and resources. In this manuscript, we find a closed formula for the \emph{Gaussian ergotropy}, i.e.~the maximum energy that can be extracted from bosonic systems governed by quadratic Hamiltonians by means of Gaussian unitaries only. This formula resembles the well-known eigenvalue-based expression for the standard ergotropy, but is instead formulated using symplectic eigenvalues.
We further prove that the Gaussian ergotropy is additive, indicating that the multiple-copy scenario does not benefit from Gaussian entangling operations. Extending our analysis to the relationship between ergotropic and entropic functions, we establish bounds linking entropic measures of Gaussianity to extractable work. Finally, we generalise our framework to open systems by studying the optimal state preparation that minimises the energy output in a Gaussian channel.

\end{abstract}

\maketitle
\tableofcontents

\section{Introduction}
Continuous variable systems~\cite{HOLEVO-CHANNELS-2,BUCCO,weedbrook12,CERF,adesso14,BARNETT-RADMORE,biblioparis}, such as quantum optical platforms, have been extensively analysed in all branches of quantum information theory, including quantum computing~\cite{Lloyd1999,Gottesman2001,Menicucci2006,Mirrahimi_2014,Ofek_nature2016,error_corr_boson,Guillaud_2019,alexander2024manufacturable,brenner2024factoringintegeroscillatorsqubit,ulysse_q_com,ulysse2,ulysse3,ulysse_5,upreti2025interplayresourcesuniversalcontinuousvariable}, quantum communication~\cite{Caves,Ralph1999,Wolf2007,TGW,PLOB,MMMM,Die-Hard-2-PRL,mele2023maximum, mele2024quantum,Die-Hard-2-PRA,mele2023optical,mele2025achievableratesnonasymptoticbosonic,mele2025nonasymptoticquantumcommunicationlossy,Sidhu,Winnel,Noh2019,Pirandola2021,Grosshans02}, quantum sensing~\cite{SGravi2,noon_pap,meyeretal,q_sensing_cv,ohliger2012continuousvariablequantumcompressedsensing,Joo2011,Facon2016,Lee2015}, and quantum learning theory~\cite{mele2024learningquantumstatescontinuous,bittel2025optimalestimatestracedistance,becker_classical_2023,gandhari2023precisionboundscontinuousvariablestate,fanizza2024efficientHamiltonianstructuretrace,fawzi2024optimalfidelityestimationbinary,oh2024entanglementenabled,wu2024efficient}. Moreover, these systems have gained significant attention in the pursuit of demonstrating a quantum advantage, particularly through boson sampling~\cite{Borealis,Zhong_2020,SupremacyReview,ulysse_4} and quantum simulation experiments~\cite{QuantumPhotoThermodynamics}. Among all continuous-variable systems, \emph{Gaussian systems}~\cite{BUCCO} stand out as some of the most significant. This is due to two key factors: first, Gaussian systems are ubiquitous in nature and in quantum optics laboratories; second, they offer a relatively straightforward mathematical framework for analysis. 
In this work, we solve the problem of how to optimally extract energy by means of Gaussian operations only, contributing to the growing literature on quantum thermodynamics with continuous-variable systems~\cite{Brown_2016,PhysRevA.101.062326,PhysRevLett.109.190502,medina2024anomalousdischargingquantumbatteries,annurev:/content/journals/10.1146/annurev-physchem-040513-103724,PhysRevE.103.012111,PhysRevResearch.5.L032010,Narasimhachar2021,PhysRevLett.124.010602,Williamson_2025,cavina2024quantum}.

To quantify the accessibility of the energy stored in a quantum state, we can use the \emph{ergotropy}~\cite{Allahverdyan_2004}, which for a state ${\rho}$ with Hamiltonian ${\hat H}$ is defined as
\begin{align} \label{eq:erg}
    & \mathcal E^{(\hat{H})} (\rho) \coloneqq E(\rho) - \min_{U} \Tr[\hat H U \rho U^\dagger] =  E (\rho) - E(\rho^{\downarrow})\,,
\end{align}
where, for an arbitrary state $\rho$, we denoted with $E(\rho) \coloneqq \Tr[\hat{H} \rho]$ its mean energy, and with $\rho^{\downarrow}$ the {\it passive state} associated to $\rho$, i.e., the lowest energy state that can be reached from $\rho$ through unitary transformations.
If we assume to have $n$ copies of the system, all prepared in the same state $\rho$, access to global entangling unitaries enhances the work extraction process. In this $n$-copy scenario, the total Hamiltonian $\hat{H}_{\mathrm{tot}}^{(n)}$ is given by the sum of the single-copy Hamiltonians, i.e.
\bb\label{eqhtot}
\hat{H}_{\mathrm{tot}}^{(n)}\coloneqq \sum_{i=1}^n\hat{H}_i\,,
\ee
where $\hat{H}_i$ acts as $\hat{H}$ on the $i$th system and acts trivially on the other systems~\cite{Lenard1978}. This leads to the definition of the {\it total ergotropy}~\cite{alicki2013entanglement, PhysRevLett.127.210601}
\begin{align} \label{m_28}
    & \mathcal E^{(\hat{H})}_\text{tot} (\rho) \coloneqq \lim_{n\to \infty} \frac{1}{n} \mathcal{E}^{(\hat{H}_{\mathrm{tot}}^{(n)})} (\rho^{\otimes n}) 
    = E(\rho) - E (\tau_{\beta^{*}}), 
\end{align}
where ${\tau_{\beta^{\ast}}}\coloneqq \frac{e^{-\beta^{\ast}\hat{H}}}{\Tr [e^{-\beta^\ast\hat{H}}]}$ is a Gibbs state with Hamiltonian $\hat{H}$ whose inverse temperature ${\beta^{*} \in \mathbb R^+}$, from now on referred to as the \textit{intrinsic inverse temperature}, is chosen to satisfy ${S (\tau_{\beta^{*}}) = S(\rho) \coloneqq -\Tr[\rho \ln \rho] }$.
\Cref{eq:erg,m_28} are genuine measures of extractable work, assuming that the work-extraction device has access to all possible unitary transformations. This freedom is typically not allowed in continuous variable systems~\cite{BUCCO}. In such systems a physically meaningful set of operations is the set of {\it Gaussian unitaries} $\mathcal{G}$. These correspond to the compositions of unitary transformations generated by quadratic Hamiltonians~\cite{BUCCO}.
In a setting where practical limitations restrict us to such transformations, it is natural to define a new measure of extractable work, the \textit{Gaussian ergotropy}:
\bb \label{m_29}
    & \mathcal E^{(\hat{H})}_G (\rho) \coloneqq E(\rho) - \min_{U \in \mathcal G} \Tr[\hat{H} U \rho U^\dagger]\,,
\ee
where the minimum is performed on the set of Gaussian unitaries $\mathcal{G}$. Moreover, throughout the rest of the paper, the Hamiltonian $\hat{H}$ is always assumed to be quadratic. Equivalently, the Gaussian ergotropy can be expressed as
\bb
    \mathcal E^{(\hat{H})}_G (\rho) =E (\rho) - E(\rho^{\downarrow}_G)\,,
\ee
where $\rho^{\downarrow}_G$ denotes the {\it Gaussian-passive state}~\cite{Brown_2016} associated to $\rho$ that, in analogy with its counterpart in~\cref{eq:erg}, corresponds to 
\begin{equation}
 \rho_G^{\downarrow} \coloneqq \bar{U} \rho \bar{U}^{\dag}\, \: \, {\rm such \, that \, \: \,} \bar{U} = \arg \min_{U\in \mathcal G}\Tr[\hat{H} U\rho U^{\dag}].
\end{equation}
Observe that in the above definitions the state \( \rho \) is not necessarily Gaussian. Moreover, if the state evolves under a quadratic Hamiltonian, the energy of the Gaussian-passive state remains constant, since the evolution is a Gaussian unitary that can be absorbed into the optimisation over $U \in \mathcal{G}$.
If this Hamiltonian coincides with $\hat{H}$, that is, the one appearing in~\cref{m_29}, the energy of $\rho$ also remains constant, ensuring that the ergotropy is a conserved quantity.

As in the standard case, we can extend the analysis to the multi-copy scenario and evaluate the advantage gained from entangling operations performed across multiple copies.
This leads to the definition of {\it Gaussian total ergotropy}:
\begin{align} \label{m_28.1}
    & \mathcal E^{(\hat{H})}_\text{G,tot} (\rho) \coloneqq \lim_{n\to \infty} \frac{1}{n} \mathcal E^{(\hat{H}_{\mathrm{tot}}^{(n)})}_G (\rho^{\otimes n})\,,
\end{align}
where $\hat{H}_{\mathrm{tot}}^{(n)}$ is defined as in~\cref{eqhtot}. 
Using these definitions, we can immediately derive some straightforward yet general bounds.
 Since the algebra of quadratic operators does not generate the whole set of unitaries, the Gaussian ergotropy is always smaller than the ergotropy. 
 We can thus define the {\it non-Gaussian work potential} as
 \begin{align} \label{eq:NGWP}
  \Delta^{(\hat{H})}(\rho) \coloneqq  \mathcal{E}^{(\hat{H})}(\rho) - \mathcal{E}^{(\hat{H})}_G(\rho) = E(\rho^{\downarrow}_G) - E(\rho^{\downarrow}) \geq 0. 
 \end{align}
Note that also the gap between the total ergotropy and the ergotropy is always positive and is called {\it bound ergotropy}~\cite{Niedenzu2019conceptsofworkin}
\bb\label{eq:bounderg}
 \mathcal{B}^{(\hat H)}(\rho) \coloneqq  \mathcal{E}^{(\hat{H})}_{\mathrm{tot}}(\rho) - \mathcal{E}^{(\hat{H})}(\rho) =  \frac{1}{\beta^*} D(\rho^{\downarrow}||\tau_{\beta^*}) \geq 0,
\ee
where $D(\rho||\sigma)\coloneqq \Tr[\rho (\ln \rho - \ln \sigma)]$ is the quantum relative entropy.
The last quantity is precisely the non-equilibrium free energy of the state $\rho^{\downarrow}$ with respect to an equilibrium temperature given by the intrinsic inverse temperature $\beta^*$.
In the scenario where multiple copies are available, one could explore the performance of local versus entangling Gaussian operations, as well as compare the effectiveness of entangling Gaussian operations with that of general entangling unitaries.
We thus introduce the {\it Gaussian bound ergotropy} $ \mathcal{B}_G$ and {\it total non-Gaussian work potential}  $ \Delta_{\mathrm{tot}} $ as 

\begin{align} \label{eq:2def}
 &  \mathcal{B}_G^{(\hat H)}(\rho) \coloneqq  \mathcal{E}^{(\hat{H})}_\text{G,tot}(\rho) - \mathcal{E}^{(\hat{H})}_G(\rho) \geq 0 ; \quad   \quad 
\Delta^{(\hat{H})}_{\mathrm{tot}}(\rho) \coloneqq    \mathcal{E}^{(\hat{H})}_{\mathrm{tot}}(\rho) - \mathcal{E}^{(\hat{H})}_\text{G,tot}(\rho) \geq 0.
\end{align}

Our main result is a simple, closed formula for the Gaussian ergotropy associated with quadratic Hamiltonians. This formula elegantly mirrors the known formula for the standard ergotropy, by replacing (roughly speaking) the eigenvalues with the \emph{symplectic} eigenvalues~\cite{BUCCO}. Specifically, the energy of the passive state associated with a state $\rho$ and a Hamiltonian $\hat{H}$ is known to be~\cite{Allahverdyan_2004}
\bb 
    \min_{U}\Tr\left[\hat{H} U\rho U^\dagger \right]=\sum_{i}\lambda_i^\uparrow(\hat{H})\lambda_i^\downarrow(\rho)\,,
\ee 
where the minimisation is over the set of all unitaries, $\{\lambda_i^\uparrow(\hat{H})\}_i$ are the eigenvalues of $\hat{H}$ ordered in increasing order, and $\{\lambda_i^\downarrow(\rho)\}_i$ are the eigenvalues of $\rho$ ordered in decreasing order. Remarkably, our main result (theorem \ref{v1_010}) establishes that the energy of the \emph{Gaussian}-passive state associated with a state $\rho$ and a quadratic Hamiltonian $\hat{H}$ is given by
\bb  \label{eq:result}
    \min_{U\in\mathcal{G}}\Tr\left[\hat{H} U\rho U^\dagger \right]=\frac12\sum_{i=1}^n d_i^\uparrow(h)d_i^\downarrow(V(\rho))\,,
\ee 
where the minimisation is restricted to the set of all Gaussian unitaries, $\{d_i^\uparrow(h)\}_i$ are the symplectic eigenvalues of the Hamiltonian matrix $h\in \mathbb R^{2 n \times 2 n}$ (see \cref{eqh}) associated with the quadratic Hamiltonian $\hat{H}$ ordered in increasing order, and $\{d_i^\downarrow(V(\rho))\}_i$ are the symplectic eigenvalues of the covariance matrix $V(\rho)$ of the state $\rho$ ordered in decreasing order.

The paper is structured as follows. In~\cref{sec:preliminaries}, we provide an overview of the theoretical tools for continuous variable systems that will be used throughout the text. In~\cref{sec:bigproof} we prove the result in Eq. \eqref{eq:result} and find the form of the optimal unitary appearing in~\cref{m_29}.
In~\cref{sec:additivity} we prove that the Gaussian ergotropy is additive. This makes the Gaussian bound ergotropy identically equal to $0$ and establishes a relation between the non-Gaussian work potential and its total version.
In~\cref{sec:entropy} we prove several bounds between ergotropic and entropic Gaussian functionals, valid for quadratic Hamiltonians.
In~\cref{sec:channels} we characterise the minimum energy at the output of a Gaussian channel. Finally, in~\cref{sec:conclusions} we draw our conclusions.

\section{Preliminaries on continuous variable systems}
\label{sec:preliminaries}
In this section, we review the relevant preliminaries regarding continuous-variable systems; for more details, we refer to~\cite{BUCCO}. A continuous variable system is a quantum system associated with the Hilbert space $L^2(\mathbb{R}^n)$ of all square-integrable complex-valued functions on $\R^n$, where $n$ denotes the number of modes~\cite{BUCCO}. On such a Hilbert space, one can define the \emph{quadrature operator vector} $\hat{\textbf{R}}\coloneqq (\hat{x}_1,\hat{p}_1,\ldots,\hat{x}_n,\hat{p}_n)$, where $\hat{x}_i$ and $\hat{p}_i$ denote the well-known position and momentum operator of the $i$th mode. The quadrature operator vector satisfies the canonical commutation relation 
\bb
[\hat{\textbf{R}},\hat{\textbf{R}}^{\intercal}]=i\,\Omega\,\mathbb{\hat{1}}\,,
\ee 
where 
\bb\label{symp_form}
        \Omega = \bigoplus_{i=1}^{n} 
\begin{pmatrix}
    0 & 1 \\ 
    -1 & 0
\end{pmatrix}
\ee
is the so-called symplectic form.

Let us define the concept of quadratic Hamiltonian, which plays a crucial role in Gaussian quantum information~\cite{BUCCO}.
\begin{Def}[(Quadratic Hamiltonian)]\label{def_quad_ham}
    An $n$-mode Hamiltonian $\hat{H}$ is said to be quadratic if it is a quadratic polynomial in the quadrature operator vector $\hat{\textbf{R}}$.
\end{Def}

\begin{Def}[(Gaussian unitary)]
    An $n$-mode unitary is said to be Gaussian if it can be written as a composition of unitaries generated by quadratic Hamiltonians. 
\end{Def}
It can be shown that the most general Gaussian unitary $G$ is of the form~\cite{BUCCO}
\bb
    G=\hat{D}_{\textbf{r}}U_S\,,
\ee
where $\hat{D}_{\textbf{r}}$ is the \emph{displacement operator}~\cite{BUCCO} associated with the amplitude $\textbf{r}\in\mathbb{R}^{2n}$ and $U_S$ is the \emph{symplectic Gaussian unitary}~\cite{BUCCO} associated with the \emph{symplectic matrix} $S$. 
\begin{Def}[(Symplectic matrix)]
    A matrix $S\in\mathbb{R}^{2n\times 2n}$ is said to be symplectic if it satisfies $S \Omega S^\intercal =\Omega$, where $\Omega$ is the symplectic form defined in~\cref{symp_form}. The set of all symplectic matrices is denoted as $\mathrm{Sp}(2n)$.
\end{Def}
Given a symplectic matrix $S$, the symplectic Gaussian unitary $U_S$ acts on the quadrature operator vector as
\bb
    U_S^\dagger \hat{\textbf{R}}U_S=S\hat{\textbf{R}}\,.
\ee
Moreover, given a vector $\textbf{r}\in\mathbb{R}^n$, the displacement operator $\hat{D}_{\textbf{r}}$ acts on the quadrature operator vector as
\bb
    \hat{D}_{\textbf{r}}^\dagger \hat{\textbf{R}}\hat{D}_{\textbf{r}}=\hat{\textbf{R}}+\textbf{r}\hat{\mathbb{1}}\,.
\ee
Let us state the following result, known as \emph{Williamson's decomposition}.
\begin{thm}[(Williamson's decomposition)] \label{v1_03}
Any strictly positive matrix $h\in\mathbb{R}^{2n\times 2n}$ can be written as
\bb\label{def_will}
    h=SDS^\intercal\,,
\ee
where $S$ is a symplectic matrix and $D$ is a diagonal matrix of the form \bb
D=\bigoplus_{j=1}^n  \left(\begin{matrix}d_j&0\\0&d_j\end{matrix}\right)\,,
\ee
where $d_1,d_2,\ldots , d_n\ge 0$ are called the symplectic eigenvalues of $h$.
\end{thm}

The first moment $\mathbf{m}(\rho)$ and the covariance matrix $V(\rho)$ of a quantum state $\rho$ are defined as 
\bb \label{cm}
\mathbf{m}(\rho)&\coloneqq \Tr\big[\hat{\textbf{R}}\,\rho\big]\,,\\
V(\rho)&\coloneqq\Tr\big[\big\{(\hat{\textbf{R}}-m(\rho)\,\hat{\mathbb{1}}),(\hat{\textbf{R}}-m(\rho)\,\hat{\mathbb{1}})^{\intercal}\big\}\rho\big]\,,
\ee
where $(\cdot)^\intercal$ denotes the transpose operation, $\{A,B\}\coloneqq AB+BA$ represents the anti-commutator~\cite{BUCCO}. Remarkably, since any covariance matrix $V$ is positive and satisfies the so-called uncertainty relation $V+i\Omega \ge0$~\cite{BUCCO}, then $V$ can be written in Williamson's decomposition with all its symplectic eigenvalues being $\ge 1$. Moreover, a symplectic Gaussian unitary $U_S$ associated with a symplectic matrix $S$ acts on the first moment and covariance matrix as~\cite{BUCCO}
\bb\label{eq_tr_m}
    \textbf{m}(U_S\rho U_S^\dagger)&=S\textbf{m}(\rho)\,,\\
    V(U_S\rho U_S^\dagger)&=SV(\rho)S^\intercal \,.
\ee
while a displacement operator $\hat{D}_{\textbf{r}}$ associated with the amplitude $\textbf{r}$ acts as~\cite{BUCCO}
\bb\label{eq_tr_V}
    \textbf{m}(\hat{D}_{\textbf{r}}\rho \hat{D}_{\textbf{r}}^\dagger)&=\textbf{m}(\rho)+\textbf{r}\,,\\
    V(\hat{D}_{\textbf{r}}\rho \hat{D}_{\textbf{r}}^\dagger)&=V(\rho) \,.
\ee
Notably, the expectation value of of a quadratic Hamiltonian onto a state depends just on the first moment and on the covariance matrix of the state, as proved in the following lemma.
\begin{lemma}[(Expectation value of a quadratic Hamiltonian)]\label{lem_energy}
    Let $\hat{H}$ be a quadratic Hamiltonian of the form
    \bb\label{eqh}
        \hat{H}\coloneqq \frac12(\hat{\textbf{R}}-\textbf{r})^\intercal h (\hat{\textbf{R}}-\textbf{r})\,,
    \ee
    where $\textbf{r}\in\mathbb{R}^{2n}$ and $h\in\mathbb{R}^{2n\times 2n}$. Then, the expectation value of $\hat{H}$ onto a state $\rho$ is given by
    \bb
        \Tr[\hat{H}\rho]=\frac{1}{4}\Tr[hV(\rho)]+ \frac12\left( \textbf{m}(\rho)-\textbf{r}\right)^\intercal h \left( \textbf{m}(\rho)-\textbf{r}\right)\,,
    \ee
    where $\textbf{m}(\rho)$ and $V(\rho)$ are the first moment and covariance matrix of $\rho$, respectively.
\end{lemma}
\begin{proof}
    It easily follows by applying the definition of first moment and covariance matrix together with~\cref{eq_tr_m} and~\cref{eq_tr_V}.
\end{proof}

The most important class of continuous-variable states is arguably the class of \emph{Gaussian states}, both from an experimental and theoretical point of view. By definition, a Gaussian state is a tensor product of Gibbs states of quadratic Hamiltonians in the quadrature operator vector~\cite{BUCCO}. A Gaussian state is uniquely identified by its covariance matrix and its first moment. In addition, the most important class of continuous-variable quantum channels is given by the class of \emph{Gaussian channels}~\cite{BUCCO}.
\begin{Def}[(Gaussian channel)]
    A Gaussian channel is a channel that maps the set of Gaussian states to itself.
\end{Def}
Since Gaussian states are uniquely identified by their first moments and convariance matrices, one can define a Gaussian channel by specifying how it acts on the first moments and covariance matrices. Following this idea, the following Lemma provides a characterization of Gaussian channels~\cite{BUCCO}.
\begin{lemma}[(Characterisation of Gaussian channels)]\label{char_gauss_ch}
    Let $\textbf{x}\in\mathbb{R}^{2n}$ and let $X,Y\in\mathbb{R}^{2n\times 2n}$ be such that the following operator inequality holds:
    \bb
        Y+i\Omega \ge i X\Omega X^\intercal\,.
    \ee
    Then, there exists a Gaussian channel $\Phi_{X,Y,\textbf{x}}$ that acts on any first moment \textbf{m} and any covariance matrix $V$ as
    \bb
        \textbf{m}&\longmapsto X\textbf{m}+\textbf{x}\,,\\
        V&\longmapsto XVX^\intercal+Y\,.
    \ee
    Moreover, any Gaussian channel is of the above form.
\end{lemma}

\section{Gaussian ergotropy}
\label{sec:bigproof}

In this section we prove the first main result of the manuscript, which provides an explicit formula for the Gaussian ergotropy in terms of the covariance matrix of a given quantum state $\rho$ assuming a quadratic Hamiltonian.

\begin{thm}[(Compact formula for the Gaussian ergotropy)] \label{v1_010}
    Let $\hat{H}$ be a quadratic Hamiltonian of the form
    \bb\label{eq:ham_quad}
        \hat{H}\coloneqq \frac12(\hat{\textbf{R}}-\textbf{r})^\intercal h (\hat{\textbf{R}}-\textbf{r})\,,
    \ee
    where $\textbf{r}\in\mathbb{R}^{2n}$, and $h\in\mathbb{R}^{2n\times 2n}$ is a strictly positive matrix. Then, the Gaussian ergotropy of a quantum state $\rho$ can be written as
    \bb
        \mathcal{E}^{(\hat{H})}_G(\rho) = \frac{1}{4}\Tr[hV(\rho)]+ \frac12\left( \textbf{m}(\rho)-\textbf{r}\right)^\intercal h \left( \textbf{m}(\rho)-\textbf{r}\right)-\frac12\sum_{j=1}^nd_j^\uparrow(h)\,d_j^{\downarrow}\!\left(  V(\rho)  \right)\,,
    \ee
    where:
    \begin{itemize}
        \item $\textbf{m}(\rho)$ is the first moment of $\rho$;
        \item $V(\rho)$ is the covariance matrix of $\rho$;
        \item $d_1^{\uparrow}(h)\le d_2^{\uparrow}(h)\le\ldots \le d_n^{\uparrow}(h) $ are the symplectic eigenvalues of $h$ ordered in increasing order;
        \item $d_1^{\downarrow}(V(\rho))\ge d_2^{\downarrow}(V(\rho))\ge\ldots \ge d_n^{\downarrow}(V(\rho)) $ are the symplectic eigenvalues of $V(\rho)$ ordered in decreasing order.
    \end{itemize}
Moreover, the Gaussian unitary $G$ that achieves the supremum in the definition of Gaussian ergotropy in~\cref{m_29} is given by the composition of the following Gaussian unitaries:
\bb\label{eq_g_opt}
    \bar{G}=D_{\textbf{r}} U_{S_2}  U_{S_1} D^\dagger_{\textbf{m}(\rho)}\,,
\ee
where 
\begin{enumerate}
    \item $D_{\textbf{m}(\rho)}$ is the displacement unitary with amplitude equal to the first moment of the state $\rho$;
    \item $U_{S_1}$ is the symplectic gaussian unitary associated with the symplectic matrix $S_1$ that puts the covariance matrix $V(\rho)$ in Williamson's decomposition with the symplectic eigenvalues ordered in decreasing order as
    \bb
        S_1V(\rho)S_1^{\intercal}= \diag\left( d_1^{\downarrow}(V(\rho)), d_1^{\downarrow}(V(\rho)), \ldots, d_n^{\downarrow}(V(\rho)),d_n^{\downarrow}(V(\rho))\right)\,,
    \ee
    \item  $U_{S_2}$ is the symplectic gaussian unitary associated with the symplectic matrix $S_2$ that puts $h$ in Williamson's decomposition with the symplectic eigenvalues ordered in increasing order as 
    \bb
        S_2^\intercal h S_2 = \diag\left( d_1^{\uparrow}(h), d_1^{\uparrow}(h), \ldots,  d_n^{\uparrow}(h), d_n^{\uparrow}(h)\right)\,,
    \ee
    \item  $D_{\textbf{r}}$ is the displacement unitary with amplitude equal to $\textbf{r}$.
\end{enumerate}
Equivalently, the Gaussian-passive state is given by $\rho_G^{\downarrow}=\bar{G}\rho\bar{G}^\dagger$.
\end{thm}

\begin{proof}
Thanks to~\cref{lem_energy}, it suffices to show that 
\bb  \label{eq:Gpass1}
    \inf_{G \in \mathcal G} \Tr[\hat{H} G \rho G^\dagger]=\frac12\sum_{j=1}^nd_j^\uparrow(h)\,d_j^{\downarrow}\!\left(  V(\rho)  \right)\,.
\ee
To this end, note that 
\bb \label{eq:Gpass2}
    \inf_{G \in \mathcal G} \Tr[\hat{H} G \rho G^\dagger]&\eqt{(i)} \inf_{G \in \mathcal G} \left[\frac{1}{4}\Tr[hV(G\rho G^\dagger)]+ \frac12\left( \textbf{m}(G\rho G^\dagger)-\textbf{r}\right)^\intercal h \left( \textbf{m}(G\rho G^\dagger)-\textbf{r}\right)\right]\\
    &\eqt{(ii)} \inf_{G \in \mathcal G} \frac{1}{4}\Tr[hV(G\rho G^\dagger)]\,\\
    &\eqt{(iii)} \inf_{S \in \mathrm{Sp}(2n)} \frac{1}{4}\Tr[hSV(\rho) S^\intercal]\,\\
    &\eqt{(iv)} \inf_{S \in \mathrm{Sp}(2n)} \frac{1}{4}\Tr[D^{\uparrow}(h)S D^{\downarrow}\!\left(V(\rho)\right) S^\intercal]\,\\
    &\eqt{(v)}  \frac{1}{4}\Tr[D^{\uparrow}(h)D^{\downarrow}\!\left(V(\rho)\right) ]\,\\
    &=  \frac{1}{2}\sum_{j=1}^nd_j^\uparrow(h)\,d_j^{\downarrow}\!\left(  V(\rho)  \right)\,.
\ee
Here, in (i), we employed~\cref{lem_energy}. In (ii), we used that: (a) $h$ is positive; (b) the set of Gaussian unitaries forms a group, so that we can minimise over $D_{\textbf{x}}G$, where $\textbf{x}$ is chosen such that $\textbf{m}(D_{\textbf{x}}G\rho G^\dagger D_{\textbf{x}}^\dagger)=\textbf{r}$; (c) the transformation of first moments and covariance matrices under displacement operators in~\cref{eq_tr_V}. In (iii), we exploited~\cref{eq_tr_m} and introduced the symplectic group $\mathrm{Sp}(2n)$. In (iv), we did the following: (a) we performed the Williamson's decompositions of $h$ and $V(\rho)$; (b) we used that the symplectic group is actually a group; (c) we introduced the diagonal matrices of symplectic eigenvalues of $h$ and $V(\rho)$:
\bb
D^\uparrow(h)&\coloneqq\bigoplus_{j=1}^n  \left(\begin{matrix}d_j^\uparrow(h)&0\\0&d_j^\uparrow(h)\end{matrix}\right)\,,\\
D^{\downarrow}\!\left(V(\rho)\right)&\coloneqq \bigoplus_{j=1}^n  \left(\begin{matrix}d_j^{\downarrow}(V(\rho))&0\\0&d_j^{\downarrow}(V(\rho))\end{matrix}\right)\,. 
\ee
Finally, in (v), we employed~\cref{lem_min_s} of~\cref{subsecf}, which establishes that the infimum is achieved by taking $S=\mathbb{1}$.

Let us now prove that the optimiser of the Gaussian ergotropy is given by the Gaussian unitary $\bar{G}$ in~\cref{eq_g_opt}. To this end, let us observe that
\bb
    \Tr[\hat{H} \bar{G} \rho \bar{G}^\dagger]&\eqt{(vi)}\frac{1}{4}\Tr[hV(\bar{G}\rho \bar{G}^\dagger)]+ \frac12\left( \textbf{m}(\bar{G}\rho \bar{G}^\dagger)-\textbf{r}\right)^\intercal h \left( \textbf{m}(\bar{G}\rho \bar{G}^\dagger)-\textbf{r}\right)\\
    &\eqt{(vii)}  \frac{1}{4}\Tr[h S_2D^{\downarrow}\!\left(V(\rho)\right)S_2^{\intercal} ]\\
    &=\frac{1}{4}\Tr[D^{\uparrow}(h)D^{\downarrow}\!\left(V(\rho)\right) ]\\
    &=\frac{1}{2}\sum_{j=1}^nd_j^\uparrow(h)\,d_j^{\downarrow}\!\left(  V(\rho)  \right)\,,
\ee
where in (vi) we employed~\cref{lem_energy} and 
in (vii) we used that
\bb 
    \textbf{m}(\bar{G}\rho \bar{G}^\dagger)&=\textbf{r}\,,\\
    V(\bar{G}\rho \bar{G}^\dagger)&= S_2D^{\downarrow}\!\left(V(\rho)\right)S_2^{\intercal}  \,.
\ee
\end{proof}

Let us apply the above result to the standard Hamiltonian $\hat{H}_0=\frac12\hat{\textbf{R}}^\intercal\hat{\textbf{R}} $, which corresponds to the choices $h=\mathbb{1}$ and $\textbf{r}=0$. By employing the above theorem, we find that the Gaussian ergotropy of a quantum state $\rho$ associated with the standard Hamiltonian $\hat{H}_0$ is given by
\bb
    \mathcal{E}^{(\hat{H}_0)}_G(\rho) = \frac{1}{4}\Tr V(\rho)+ \frac12 \|\textbf{m}(\rho)\|^2-\frac12\sum_{j=1}^n\,d_j\!\left(  V(\rho)  \right)\,,
\ee
where $d_1\!\left(  V(\rho)  \right),\ldots, d_n\!\left(  V(\rho)  \right)$ denote the symplectic eigenvalues of $V(\rho)$.

\subsection{Minimisation over symplectic matrices}\label{subsecf}
This subsection is devoted to the proof of the following lemma, which is one of the key steps of the proof of the closed formula for the Gaussian ergotropy.
\begin{lemma}\label{lem_min_s}
    Let $n\in\mathbb{N}$ and let
    \bb \label{v1_04}
D_1^\uparrow &\coloneqq\bigoplus_{j=1}^n  \left(\begin{matrix}\alpha_j^\uparrow&0\\0&\alpha_j^\uparrow\end{matrix}\right)\,,\\
D_2^{\downarrow}\!&\coloneqq \bigoplus_{j=1}^n  \left(\begin{matrix}\beta_j^{\downarrow}&0\\0&\beta_j^{\downarrow}\end{matrix}\right)\,,
\ee
where $0\le \alpha_1^\uparrow\le \ldots \le \alpha_n^\uparrow$ and $\beta_1^\downarrow\ge \ldots \ge \beta_n^\downarrow\ge0$. Then, it holds that
\bb \label{v1_07}
    \min_{S\in\mathrm{Sp}(2n)}\Tr[D_1^\uparrow S D_2^\downarrow S^\intercal]=\Tr[D_1^\uparrow  D_2^\downarrow ]\,,
\ee
where the infimum is performed over the set $\mathrm{Sp}(2n)$ of symplectic matrices.
\end{lemma}
The above Lemma is equivalent to the following one, where we switch the ordering convention in the definition of symplectic matrices.
\begin{lemma}\label{lem_min_soff} 
    Let $n\in\mathbb{N}$ and let
    \bb \label{v1_04}
D_1^\uparrow &\coloneqq\diag(  \alpha_1^\uparrow,\ldots \alpha_n^\uparrow,  \alpha_1^\uparrow,\ldots \alpha_n^\uparrow )\\
D_2^{\downarrow}\!&\coloneqq \diag( \beta_1^\downarrow, \ldots , \beta_n^\downarrow,\beta_1^\downarrow, \ldots , \beta_n^\downarrow)\,,
\ee
where $0\le \alpha_1^\uparrow\le \ldots \le \alpha_n^\uparrow$ and $\beta_1^\downarrow\ge \ldots \ge \beta_n^\downarrow\ge0$. Then, it holds that
\bb \label{v1_07}
    \min_{S\in\mathrm{Sp}'(2n)}\Tr[D_1^\uparrow S D_2^\downarrow S^\intercal]=\Tr[D_1^\uparrow  D_2^\downarrow ]\,,
\ee
where the infimum is performed over the set $\mathrm{Sp}'(2n)$ of symplectic matrices "with respect to the xp-ordering". That is, $\mathrm{Sp}'(2n)$ is the set of all $2n\times 2n$ matrices such that $SJS^\intercal=J$, where
    \bb 
        &J= \begin{bmatrix} 0  & \mathbb{1} \\ -\mathbb{1} & 0 \end{bmatrix}\,.
    \ee
\end{lemma}
Before presenting the proof of the above lemma, let us introduce some preliminary results.

    \begin{Def}[(Doubly stochastic matrix)] \label{a_24}
        An ${n\times n}$ matrix ${B = (b_{ij})}$ is \textit{doubly stochastic} if 
        \begin{align} \label{a_22}
            & b_{ij} \geq 0,\quad \forall\, i,j = 1,\,\dots,\, n\,,
        \end{align}
        and 
        \begin{align} \label{a_23}
             \sum_{i=1}^n b_{ij}&=1,\quad\forall\, j = 1,\, \dots,\, n\,\\
            \sum_{j=1}^n b_{ij}&=1\,,\quad\forall\, i = 1,\, \dots,\, n.
        \end{align}
    \end{Def}
    From the above definition, it is easy to see that if a matrix ${B}$ is doubly stochastic, then its transpose ${B^\intercal}$ is doubly stochastic, too.
    According to Birkhoff's Theorem~\cite{puntanen2011inequalities}, $n\times n$ doubly stochastic matrices are convex combinations of permutation matrices ${\{P_\pi\}}$, i.e.~$n\times n$ binary matrices with exactly one entry of ${1}$ in each row and column, and all other entries ${0}$.
    Mathematically, for any $n\times n$ doubly stochastic matrix ${B}$ there exists a probability distribution $\{\lambda_{\pi}\}$ such that~\cite{puntanen2011inequalities}
    \begin{align} \label{a_21}
        & B = \sum_{\pi}\lambda_\pi P_\pi \,,
    \end{align}
    where the sum is over all permutations $\pi$ of $n$ elements. Let us state the following useful result proved in~\cite[Theorem 6]{10.1063/1.4935852}.
    \begin{lemma}[(Theorem 6 of~\cite{10.1063/1.4935852})]\label{lemma_dio_bathia}
   Let $S$ be an $2n\times 2n$ symplectic matrix and let ${A = (A_{ij})}$, ${B = (b_{ij})}$,  ${C = (c_{ij})}$ , ${G = (g_{ij})}$ be $n\times n$ matrices such that
    \bb \label{v1_05}
        &S = \begin{bmatrix} A  & B \\ C & G \end{bmatrix}\,.
    \ee
    Moreover, let ${\tilde{S} = (\tilde{s}_{ij})}$  be an $n\times n$ matrix defined as
    \bb \label{v1_01}
        & \tilde s_{ij}\coloneqq \frac{1}{2}(a^2_{i j} + b^2_{i j} + c^2_{i j} + g^2_{i j})\,.
    \ee
    Then, there exists an $n\times n$ doubly stochastic matrix ${Q}=(q_{ij})$ such
    \bb \label{v1_02}
        & \tilde s_{ij}\geq q_{ij}
    \ee
    for all $i,j =1, \dots,n$.
    \end{lemma}
We are now ready to prove~\cref{lem_min_soff}.

\begin{proof}[Proof of~\cref{lem_min_soff}]
    By a direct calculation and employing the notation introduced in~\cref{lemma_dio_bathia}, one can easily verify that
    \bb
        \Tr[D_1^\uparrow S D_2^\downarrow S^\intercal]= 2\sum_{i,j=1}^n \alpha_i^\uparrow \tilde s_{ij} \beta^{\downarrow}_j \,.
    \ee
    Moreover, thanks to~\cref{lemma_dio_bathia}, there exists a probability distribution $(p_\pi)$ such that
    \bb
        \tilde{s}_{ij}\ge \sum_\pi\lambda_\pi (P_\pi)_{ij}\,\qquad\forall\,i,j=1,\ldots,n\,,
    \ee
    where the sum is over all the permutations $\pi$ of $n$ elements. Hence, we obtain that
    \bb
        \Tr[D_1^\uparrow S D_2^\downarrow S^\intercal]&\ge 2\sum_{\pi}\lambda_\pi\sum_{i,j=1}^n \alpha_i^\uparrow (P_k)_{ij} \beta^{\downarrow}_j\\
        & \geqt{(i)} 2\sum_{\pi}\lambda_\pi\sum_{i,j=1}^n \alpha_i^\uparrow \beta^{\downarrow}_j\\
        &=2\sum_{i,j=1}^n \alpha_i^\uparrow \beta^{\downarrow}_j\\
        &=\Tr[D_1^\uparrow D_2^\downarrow]\,,
    \ee
    where in (i) we employed the rearrangement inequality. Consequently, since the identity is a symplectic matrix, we conclude that
    \bb
        \min_{S\in\mathrm{Sp}'(2n)}\Tr[D_1^\uparrow S D_2^\downarrow S^\intercal]=\Tr[D_1^\uparrow  D_2^\downarrow ]\,.
    \ee

\end{proof}
\subsection{Comment on the anti-ergotropy}
Analogously to the finite-dimensional setting, one may be interested in the analysis of the so-called \emph{anti-ergotropy}~\cite{PhysRevLett.131.030402}. The anti-ergotropy, denoted as $\mathcal{A}^{(\hat{H})}(\rho)$, is defined as the maximum energy --- measured by a Hamiltonian $\hat{H}$ --- that can be injected into a quantum system initialised in the state $\rho$ by acting with unitary operations only:
\bb
\mathcal{A}^{(\hat{H})}(\rho)\coloneqq\sup_{U}\left(\Tr[\hat{H}U\rho U^\dagger]-\Tr[\hat{H}\rho]\right)\,,
\ee
where the supremum is performed over all possible unitaries. However, since the spectrum of bosonic quadratic Hamiltonians is unbounded, it follows that the anti-ergotropy is always infinite for bosonic quadratic Hamiltonians.

What about if we restrict such an optimisation to the smaller, yet practically relevant, subset of Gaussian unitaries? Following this idea, analogously to what we did in the previous section, we may define \emph{Gaussian anti-ergotropy}. The Gaussian anti-ergotropy, denoted as $\mathcal{A}_G^{(\hat{H})}(\rho)$, is defined as the maximum energy --- measured by an Hamiltonian $\hat{H}$ --- that one can inject into a quantum system initialised in the state $\rho$ by acting with \emph{Gaussian} unitaries only:
\bb
\mathcal{A}^{(\hat{H})}(\rho)\coloneqq\sup_{U\in\mathcal{G}}\left(\Tr[\hat{H}U\rho U^\dagger]-\Tr[\hat{H}\rho]\right)\,,
\ee
where the supremum is restricted to the set $\mathcal{G}$ of all possible Gaussian unitaries. However, it is simple to observe that the Gaussian anti-ergotropy is still infinite for bosonic quadratic Hamiltonians. This is so because one could take, e.g., a suitable Gaussian unitary with infinite squeezing or infinite displacement. In conclusion, both the anti-ergotropy and the Gaussian anti-ergotropy are infinite for bosonic systems governed by quadratic Hamiltonians. 
This makes the analysis of the anti-ergotropy of bosonic systems less interesting than that of its finite-dimensional counterpart~\cite{PhysRevLett.131.030402}.

\section{Additivity of Gaussian ergotropy}
\label{sec:additivity}

In this section, we prove that for a quadratic Hamiltonian $\hat{H}$ the Gaussian ergotropy is additive. This proof is fundamental to give a complete characterisation of the quantities defined in~\cref{eq:2def}.

\begin{lemma}[(Additivity of ${\mathcal E^{(\hat{H})}_G}$ for quadratic Hamiltonians)] \label{m_22}
    Given a Hamiltonian of the form in~\cref{eq:ham_quad}, the Gaussian ergotropy ${\mathcal E_G (\rho)}$ is additive, i.e.,
    \begin{align} \label{m_21}
        & \mathcal E_G^{(\hat{H}^{(n)}_{\mathrm{tot}})} (\rho^{\otimes n}) = n\, \mathcal E_G^{(\hat{H})} (\rho).
    \end{align}
    Thus, it saturates the general lower bound ${\mathcal E^{(\hat{H}^{(n)}_{\mathrm{tot}})} (\rho^{\otimes n} ) \geq n \mathcal E^{(\hat{H})} (\rho)}$.
\end{lemma}

\begin{proof}
    Consider the tensor product Gaussian states ${\rho^{\otimes n}}$ with a total Hamiltonian given by the sum of local Hamiltonians. 
    In the formalism of continuous variable systems, such Hamiltonian can be expressed as
    \bb {\hat{H}^{(n)} = \frac{1}{2} (\hat {\mathbf R}^{(n)} - \textbf r^{(n)})^\intercal h^{(n)} (\hat{\mathbf R}^{(n)} - \textbf r^{(n)}) },
    \ee where ${\textbf r^{(n)} = \bigoplus_{k=1}^n \textbf r}$, ${\hat{\mathbf R}}^{(n)} = \bigoplus_{k=1}^n \hat {\mathbf R}$ and ${h^{(n)} = \bigoplus_{k=1}^n h}$.
    After we apply \cref{v1_010}, the Gaussian ergotropy of this system becomes
    \begin{align} \label{a_40}
        & \mathcal E_G^{(\hat{H}^{(n)}_{\mathrm{tot}})} (\rho^{\otimes n}) = E (\rho^{\otimes n}) - \min_{U\in\mathcal{G}} \Tr[\hat H U \rho U^\dagger] = E (\rho^{\otimes n}) - \sum_{j=1}^n \nu_j ^\downarrow (h^{(n)}) \nu_j ^\uparrow (V(\rho^{\otimes n})).
    \end{align}
To prove additivity, we show that both the energy and the energy of the Gaussian-passive state (the first and second terms on the right-hand side of~\cref{a_40}) are additive.
The first statement follows directly, as the Hamiltonian $\hat{H}$ is a sum of local operators and the state is factorised:
    \bb {E(\rho^{\otimes n}) = \Tr[\hat H_{\mathrm{tot}}^{(n)} \rho^{\otimes n}] = \sum_{k=1}^n \frac{1}{2} \Tr[(\hat{\mathbf R} - \textbf r )^\intercal h (\hat{\mathbf R} - \textbf r) \rho ] = n E(\rho) }. \ee
To prove the second statement we use that the symplectic eigenvalues of a direct sum of operators are the direct sum of the symplectic eigenvalues of the respective operators, that is  
\begin{equation} 
d_j^{\downarrow/\uparrow} \Big(\bigoplus_{k=1}^n A_k\Big) =  \bigoplus_{k=1}^n d_j^{\downarrow/\uparrow}(A_k).  \end{equation}
Finally, the covariance matrix of a product of Gaussian states is the direct sum of the covariance matrices of the states
\begin{equation}
    {V(\rho^{\otimes n}) = \bigoplus_{k=1}^n V(\rho) }.
\end{equation}
 Combining these results, we have  
    \begin{align} \label{a_41}
        & \mathcal E_G^{(\hat{H}^{(n)}_{\mathrm{tot}})} (\rho^{\otimes n}) = E (\rho^{\otimes n}) - \sum_{j=1}^n d_j ^\downarrow (h^{(n)}) d_j ^\uparrow (V(\rho^{\otimes n})) = n E (\rho) - \sum_{k=1}^n \sum_{j=1}^n d_j ^\downarrow (h) d_j ^\uparrow (V(\rho)) \nonumber\\
        & = n \Big( E(\rho) - \sum_{j=1}^n d_j ^\downarrow (h) d_j ^\uparrow (V(\rho)) \Big) = n \mathcal E_G (\rho),
    \end{align}
that concludes the proof.
\end{proof}
Let us now recall the definition of total ergotropy given in~\cref{m_28.1}. A direct consequence of~\cref{m_22} is the following%
\begin{cor} [(Gaussian bound ergotropy is identically zero)]
Given a quadratic Hamiltonian $\hat{H}$ and a generic quantum state $\rho$, the total Gaussian ergotropy defined in~\cref{m_28.1} is equal to the Gaussian ergotropy, that is
    \begin{equation}
        \mathcal{E}^{(\hat{H})}_\text{G, tot}(\rho) = \mathcal{E}^{(\hat{H})}_{G}(\rho).
    \end{equation}
    In addition, the Gaussian bound ergotropy and the total non-Gaussian work potential defined in~\cref{eq:2def} satisfy
    \begin{equation} \label{eq:closure}
        \mathcal{B}_G^{(\hat H)}(\rho) = 0, \quad  \quad \Delta_{\mathrm{tot}}^{(\hat H)}(\rho) =  \Delta^{(\hat{H})}(\rho) +  \mathcal{B}^{(\hat H)}(\rho),
    \end{equation}
    where $\Delta^{(\hat{H})}(\rho)$ and $\mathcal{B}^{(\hat H)}(\rho)$ are given in~\cref{eq:NGWP} and~\cref{eq:bounderg}, respectively.
\end{cor}
An interesting consequence of this corollary is that, when dealing with quadratic Hamiltonians and Gaussian operations, operations that create correlations between different copies of the state provide no advantage.
This is fundamentally different, for example, from what is observed in the case of incoherent operations~\cite{PhysRevLett.125.180603}, where the 
n-copies framework outperforms the single-copy scenario due to the accumulation of correlations.

\section{Entropic relations for the Non-Gaussian Work Potential}
\label{sec:entropy}

The ergotropy and the non-equilibrium free energy are closely related thermodynamic resources. In particular, the gap between total ergotropy and ergotropy (the bound ergotropy in~\cref{eq:bounderg}) coincides exactly with the non-equilibrium free energy with temperature $\beta^*$.
It is natural to ask whether a connection between entropic and ergotropic quantities emerges by examining another gap, the non-Gaussian work potential $\Delta^{(\hat{H})}(\rho)$. 
This section aims to:  
(i) establish a connection between ergotropic and entropic quantities, specifically linking the total non-Gaussian work potential to the relative entropy between the Gaussian passive state and the thermal state with inverse temperature matching the entropy of the initial state (see~\cref{m_24});  
(ii) derive general bounds for the non-Gaussian work potential in terms of entropic measures of non-Gaussianity (see~\cref{lemma:eqdelrho}); and  
(iii) provide bounds on ergotropy and total ergotropy in terms of their Gaussian counterparts (see~\cref{theo:ent-bounds} and \cref{cor:totb}).

\begin{thm}[(Explicit form of $\Delta_{\mathrm{tot}}^{(\hat H)}(\rho)$)] \label{m_24}
    For any quantum state $\rho$ and quadratic Hamiltonian $\hat{H}$ of the form given in~\cref{eq:ham_quad} we have the total non-Gaussian work potential can be expressed as
    \bb \label{m_30}
        & \Delta_{\mathrm{tot}}^{(\hat H)}(\rho) = E(\rho_G^{\downarrow}) - E (\tau_{\beta^*}) = \frac{1}{\beta^{*}} D (\rho_G^{\downarrow}||\tau_{\beta^{*}})\,,
    \ee
    where $\beta^*$ is the inverse temperature such that $S(\rho) = S(\tau_{\beta^*})$
     (see below~\cref{m_28}) and $\rho_{G}^{\downarrow}$ is the Gaussian-passive state of $\rho$.
    The relative entropy in \cref{m_30} can be evaluated analytically:
    \bb \label{m_31}
        & \Delta_{\mathrm{tot}}^{(\hat H)}(\rho) =\frac{1}{2} \sum_{j=1}^n d_j ^\uparrow (h) \Bigg( d_j ^\downarrow (V(\rho)) -  \coth\Big( \frac{\beta^* d_j^\uparrow (h)}{2} \Big ) \Bigg),
    \ee
    where $d^{\uparrow}_j(h)$ are the symplectic eigenvalues of $h$ (see~\cref{eq:ham_quad}) arranged in increasing order and  $d^{\downarrow}_j(V(\rho))$ are the symplectic eigenvalues of the covariance matrix of $\rho$ arranged in decreasing order. 
\end{thm}

\begin{proof}
    Using the definitions of total non-gaussian work potential and the result given in~\cref{eq:closure}, we have
    \begin{align} \label{a_26}
        &\Delta_{\mathrm{tot}}^{(\hat H)}(\rho) = \mathcal B^{(\hat H)}(\rho) +  \Delta^{(\hat H)} (\rho) =\mathcal E_\text{tot}^{(\hat{H})} (\rho) - \mathcal E_G^{(\hat{H})} (\rho) 
        = E(\rho) - E (\tau_{\beta^{*}}) - E(\rho) + E(\rho_G^{\downarrow})  =  E(\rho_G^{\downarrow} ) - E(\tau_{\beta^{*}}). 
    \end{align}
We want to prove that the last member of the chain of equalities above is equal to $(\beta^*)^{-1} D(\rho_G^{\downarrow} || \tau_{\beta^{*}})$.
We use
\begin{equation}
    D(\rho_G^{\downarrow} || \tau_{\beta^*}) =  \Tr[\rho_G^{\downarrow} \ln \rho_G^{\downarrow}] - \Tr[\rho_G^{\downarrow} \ln \tau_{\beta^*}] = - S(\rho_G^{\downarrow}) + \beta^* \Tr[\rho_G^{\downarrow} \hat{H}] + \ln Z_{\beta^*}.
\end{equation}
Using the well-known relation for the logarithm of the partition function 
\begin{equation}
    - \ln Z_{\beta^*} = \beta^* \Tr[\hat{H}\tau_{\beta^*}] - S(\tau_{\beta^*}).  
\end{equation}
We obtain 
\begin{equation} \label{eq:2ent}
 D(\rho_G^{\downarrow}|| \tau_{\beta^*}) =  - S(\rho_G^{\downarrow}) + \beta^* \Tr[\rho_G^{\downarrow} \hat{H}] - \beta^* \Tr[\hat{H}\tau_{\beta^*}] + S(\tau_{\beta^*}).  
\end{equation}
Since the passive state is obtained from the original state only through unitary Gaussian operations, we have $S(\rho) = S(\rho_G^{\downarrow})$.
By definition of $\beta^{*}$ we also have $S(\rho) = S(\tau_{\beta^*})$. This allows us to simplify the two entropies in~\cref{eq:2ent}
and obtain

\begin{equation} \label{eq:2ent2.0}
 D(\rho_G^{\downarrow} || \tau_{\beta^*}) =  \beta^* \Tr[\rho_G^{\downarrow} \hat{H}] - \beta^* \Tr[\hat{H}\tau_{\beta^*}],
\end{equation}
that plugged inside~\cref{a_26} concludes the first part of the proof.
To show~\cref{m_31}
 we take ${E(\rho_G^{\downarrow})}$ as in~\cref{eq:Gpass1,eq:Gpass2}, while ${E (\tau_{\beta^*})}$ is given by 
     \begin{align} \label{m_34}
         & E (\tau_{\beta^*})= \frac{1}{4} \Tr[h V(\tau_{\beta^*})] + \frac{1}{2}(\mathbf m (\tau_{\beta^*}) - \mathbf r)^\intercal h (\mathbf m (\tau_{\beta^*}) -  \mathbf r).
     \end{align}
     For the thermal state, we have~\cite{BUCCO}
     \begin{align} \label{m_36}
         & \mathbf m (\tau_{\beta^*}) = \mathbf r,\quad V(\tau_{\beta^*}) = S \Big( \bigoplus_{j=1}^n \nu_j^\uparrow \mathbb 1_2 \Big) S^\intercal,
     \end{align}
     where
     \begin{align} \label{m_37}
         & \nu_j^\uparrow = \frac{1+e^{-\xi_j^\uparrow}}{1-e^{-\xi_j^\uparrow}}
         =\coth \Big(\frac{\xi_j ^\uparrow}{2}\Big) \geq 1,\quad \xi_j^\uparrow = \beta^* d_j^\uparrow (h).
     \end{align}
     Here, we have used the Williamson's decomposition to write
         $ h = S^{-\intercal} d^\uparrow (h) S^{-1}$
     where ${d_1^\uparrow (h)  \leq \dots \leq d_n^\uparrow(h) }$. 
     From this result, we obtain 
     \begin{align} \label{m_35}
         & E (\tau_{\beta^*})= \frac{1}{4} \sum_{j=1}^{2n} d_j ^\uparrow (h) \Big( \bigoplus_{l=1}^n \nu_l^\uparrow \mathbb 1_2 \Big)  = \frac{1}{2} \sum_{j=1}^n d_j^\uparrow (h) \coth\left(\frac{\beta^* d_j^\uparrow (h)}{2}\right).
     \end{align}
    Combining the above equation with the expression in~\cref{eq:Gpass2} for the Gaussian-passive state, we obtain~\cref{m_31}.
\end{proof}

We can establish another strong result that provides an equivalent formulation of the non-Gaussian work potential, analogous to the one previously proven for coherent and incoherent contributions to ergotropy~\cite{PhysRevLett.125.180603}.
To proceed further, it is essential to introduce the following entropic non-Gaussianity measure~\cite{Marian2013,PhysRevA.82.052341}
\begin{align} \label{eq:mu}
    & \mu(\rho) \coloneqq \inf_{\sigma \in \mathfrak{G}} D(\rho || \sigma) =S(\delta(\rho)) - S(\rho), 
\end{align}
where $\mathfrak{G}$ is the set of Gaussian states and $\delta(\rho)$ is the Gaussian state with the same first and second moments as $\rho$. We will also refer to $\delta(\rho)$ as the {\it Gaussianification} of $\rho$. 
\begin{lemma}[(Equivalent formulation for $\Delta^{(\hat{H})}(\rho)$)]\label{lemma:eqdelrho}
For any quantum state $\rho$ and quadratic Hamiltonian $\hat{H}$ we have that the non-Gaussian work potential can be expressed as
\begin{equation}  \label{eq:2Ds2a}
\beta \Delta^{(\hat{H})}(\rho) 
= \mu(\rho) + D(\delta_G^{\downarrow}(\rho)|| \tau_{\beta})  
- D(\rho^{\downarrow}|| \tau_{\beta} ) .
\end{equation}
where $\beta$ is an inverse temperature (that can be chosen freely), $\rho^{\downarrow}$ is the passive state of $\rho$ and $\delta^{\downarrow}_G(\rho)$ is the Gaussian-passive state associated to its Gaussianification $\delta(\rho)$.
\end{lemma}
\begin{proof}
For a generic state $\rho$, the Gaussian gap can be written as
\begin{equation}
 \label{eq:Engapp}   \beta \Delta^{(\hat{H})}(\rho) \coloneqq \beta \big[ \mathcal{E}^{(\hat{H})}(\rho) - \mathcal{E}^{(\hat{H})}_{G}(\rho) \big]  = 
      \beta \Tr[\hat{H} (\rho_G^{\downarrow} - \rho^{\downarrow}) ] = \beta \Tr[\hat{H} (\delta_G^{\downarrow}(\rho) - \rho^{\downarrow}) ].
\end{equation}
In the last equality, we used that
\begin{align} \label{m_32}
    & \Tr[\hat{H} \rho_G^{\downarrow}]  =  \min_{U \in \mathcal{G}} \big\{ \Tr[\hat{H} U \rho U^{\dag}] \big\}  = \min_{U \in \mathcal{G}} \big\{ \Tr[U^{\dag } \hat{H} U \rho ] \big\} =  \min_{U \in \mathcal{G}} \big\{ \Tr[U^{\dag } \hat{H} U \delta(\rho) ] \big\} =  \Tr[\hat{H} \delta_G^{\downarrow}(\rho)],
\end{align}
since the average of the quadratic observable $U^{\dag} \hat{H} U$ on the state $\rho$ is equal to its average over the Gaussian state with the same first and second moments $\delta(\rho)$.
We have

\begin{equation}  \label{eq:2Ds}
\beta \Delta^{(\hat{H})}(\rho) = \beta \Tr[\hat{H} (\delta_G^{\downarrow}(\rho) - \rho^{\downarrow}) ]
 = D(\delta_G^{\downarrow}(\rho)|| \tau_{\beta})  + S(\delta_G^{\downarrow}(\rho)) 
-\big[ D(\rho^{\downarrow} || \tau_{\beta} ) + S(\rho^{\downarrow}) \big].
\end{equation}

To prove the equality above, we expand the definition of relative entropy to obtain: 
\begin{align}
   & D(\delta_G^{\downarrow}(\rho)|| \tau_{\beta}) + S(\delta_G^{\downarrow}(\rho)) = \beta \Tr[\hat{H} \delta_G^{\downarrow}(\rho)] + \ln Z_{\beta}, \\
    &  D(\rho^{\downarrow}|| \tau_{\beta}) + S(\rho^{\downarrow}) = \beta \Tr[\hat{H} \rho^{\downarrow}] + \ln Z_{\beta}.
\end{align}
From~\cref{eq:2Ds}, using the identities $ S(\rho^{\downarrow}) = S(\rho) $, $S(\delta_G^{\downarrow} (\rho)) = S(\delta(\rho)) $ and ~\cref{eq:mu} we obtain
\begin{equation}  \label{eq:2Ds2}
\beta \Delta^{(\hat{H})}(\rho) 
= \mu(\rho) + D(\delta_G^{\downarrow}(\rho)|| \tau_{\beta})  
- D(\rho^{\downarrow}|| \tau_{\beta} )\, .
\end{equation}
\end{proof}

\cref{lemma:eqdelrho} provides a complete characterisation of the non-Gaussian work potential in terms of entropic quantities. Since relative entropy and $\mu(\rho)$ are positive-definite functions, the form given in~\cref{eq:2Ds2a} is particularly useful to establish bounds for $\Delta^{(\hat{H})}$. Notably, deriving bounds in terms of functions of Gaussian states and quadratic Hamiltonians is especially valuable, as these quantities are computationally tractable. The following Theorem goes exactly in this direction.

\begin{lemma}[(Bounds for $ \Delta^{(\hat{H})}(\rho)$\label{theo:ent-bounds})]
Let $\rho$ be a generic quantum state, and let $\hat{H}$ be a generic quadratic Hamiltonian. Then, the non-Gaussian work potential can be bounded as
\begin{equation}
\mu(\rho) -  D(\rho^{\downarrow}|| \tau_{\beta} ) \le \beta \Delta^{(\hat{H})}(\rho) \leq S(\delta(\rho))-S(\delta(\rho^{\downarrow}))+ D(\delta_G^{\downarrow}(\rho)|| \tau_{\beta}) \label{eq:theoD}
\end{equation} 
where $\beta$ is an inverse temperature (that can be chosen freely), $\delta^{\downarrow}_G(\rho)$ is the Gaussian-passive state associated with $\delta(\rho)$, $\delta(\rho^{\downarrow})$ is the Gaussianification of the passive state $\rho^{\downarrow}$. 
\end{lemma}

\begin{proof} 
From~\cref{lemma:eqdelrho} we can obtain the upper bound in the following way. Let us first notice that $\tau_\beta$ is a Gaussian state, so it follows that 
\bb 
D(\rho^{\downarrow}||\tau_B) \ge \min_{\sigma\in\mathfrak{G}} D(\rho^{\downarrow}||\sigma) = \mu(\rho^{\downarrow}).
\ee 
By replacing this result in~\cref{eq:2Ds2} we obtain the upper bound
\bb 
\beta \Delta^{(\hat{H})}(\rho) \le S(\delta(\rho)) - S(\delta(\rho^{\downarrow}))+ D(\delta_G^{\downarrow}(\rho)|| \tau_{\beta}),
\ee 
where we used that $\mu(\rho)=S(\delta(\rho))-S(\rho)$. 
In addition, using the positivity of the relative entropy in~\cref{eq:2Ds2}, one we can obtain 
\bb \label{eq:otherbound}
\mu(\rho) -  D(\rho^{\downarrow}|| \tau_{\beta} ) \le\beta \Delta^{(\hat{H})}(\rho)\le 
 \mu(\rho) + D(\delta_G^{\downarrow}(\rho)|| \tau_{\beta}),
\ee 
from which the lower bound in \cref{eq:theoD} is derived.
Notice that the upper bound in \cref{eq:otherbound} is instead always looser than the one provided in \cref{eq:theoD} since $S(\delta(\rho^{\downarrow})) \geq S(\rho^{\downarrow}) = S(\rho)$.
\end{proof}

From the above Lemma we can obtain a bound
on the total ergotropy in terms of functions that quantify two key resources in the work extraction process: non-Gaussianity (that is, given in terms of the function $\mu(\rho)$) and extractable work under unitaries (quantified by $\mathcal{E}^{\hat{H}}(\rho)$).
\begin{cor} [(Bound on the total ergotropy)] \label{cor:totb}
Let $\rho$ be a generic quantum state with quadratic Hamiltonian $\hat{H}$. The total ergotropy defined in~\cref{m_28} satisfies
\begin{align} \label{a_55}
    & \mathcal E^{(\hat{H})}_\text{tot} (\rho) \geq \frac{1}{\beta^*} \mu(\rho) + \mathcal E^{(\hat{H})}_G (\rho),
\end{align}
where $\beta^*$ is the intrinsic inverse temperature.
The bound is saturated if and only if $\rho^{\downarrow} = \tau_{\beta}$ for some $\beta \geq 0$. 
In that case $\beta = \beta^*$, $\mu = 0 $, and $\mathcal E^{(\hat{H})}_\text{tot} (\rho) = \mathcal E^{(\hat{H})}_G $.
\end{cor}
\begin{proof}
    If we choose $\beta= \beta^*$ inside~\cref{eq:otherbound}
we have
\begin{equation}
     \mu(\rho) - \beta^* \mathcal{B}^{(\hat H)}(\rho) \leq \beta^{*} \Delta^{(\hat{H})}(\rho) \leq \mu(\rho) + D(\delta_G(\rho)|| \tau_{\beta^*})
\end{equation}
and the left inequality becomes, using the definition of bound ergotropy
\begin{align} \label{a_55}
    & \mathcal E^{(\hat{H})}_\text{tot} (\rho) \geq \frac{1}{\beta^*} \mu(\rho) + \mathcal E_G (\rho),
\end{align}
that is the main statement of~\cref{cor:totb}.
\end{proof}

\section{Minimum energy at the output of a Gaussian channel}\label{sec:channels}
In this section, we answer the following simple-looking question: what is the minimum energy at the output of a Gaussian channel?

\begin{thm}[(Minimum energy at the output of a Gaussian channel)]
    Let $\Phi_{X,Y,\textbf{x}}$ be a Gaussian channel as in~\cref{char_gauss_ch}, and let us assume that $X$ is invertible. Let $\hat{H}$ be a quadratic Hamiltonian of the form
    \bb
        \hat{H}\coloneqq \frac12(\hat{\textbf{R}}-\textbf{r})^\intercal h (\hat{\textbf{R}}-\textbf{r})\,,
    \ee
    where $\textbf{r}\in\mathbb{R}^{2n}$, and $h\in\mathbb{R}^{2n\times 2n}$ is a strictly positive matrix. Let $d_1(X^\intercal h X),\ldots, d_n(X^\intercal h X)$ be the symplectic eigenvalues of $X^\intercal h X$. Then, the minimum energy at the output of $\Phi_{X,Y,\textbf{x}}$ is given by
    \bb
        \min_{\rho}\Tr[\hat{H}\Phi_{X,Y,\textbf{x}}(\rho)]=\frac12
        \sum_{i=1}^nd_i(X^\intercal h X)+\frac14\Tr[hY]\,,
    \ee
    where the minimum is taken over all the input states $\rho$. Moreover, a state $\rho$ achieving such a minimum is the Gaussian state with first moment and covariance matrix 
    \bb
         \textbf{m}(\rho)&= X^{-1}(\textbf{r}-\textbf{x})\,,\\
         V(\rho)&= (S^{-1})^\intercal S^{-1}\,,
    \ee
    where $S$ is a symplectic matrix that puts $X^\intercal h X$ in Williamson's decomposition (i.e.~$X^\intercal h X=S DS^\intercal$).
\end{thm}
\begin{proof}
    It holds that
    \bb\label{eq_ters}
        \Tr[\hat{H}\Phi_{X,Y,\textbf{x}}(\rho)]&\eqt{(i)}\frac{1}{4}\Tr[hV(\Phi_{X,Y,\textbf{x}}(\rho))]+ \frac12\left( \textbf{m}(\Phi_{X,Y,\textbf{x}}(\rho))-\textbf{r}\right)^\intercal h \left( \textbf{m}(\Phi_{X,Y,\textbf{x}}(\rho))-\textbf{r}\right)&\\
        &\eqt{(ii)}\frac14\Tr[hY]+\frac{1}{4}\Tr[hXV(\rho)X^\intercal]+ \frac12\left( X\textbf{m}(\rho)+\textbf{x}-\textbf{r}\right)^\intercal h \left( X\textbf{m}(\rho)+\textbf{x}-\textbf{r}\right)\,,
    \ee
    where in (i) we employed~\cref{lem_energy} and in (ii) we used~\cref{char_gauss_ch}. One can easily show that the last term in~\cref{eq_ters} is minimised --- and vanishes --- by choosing 
    \bb
        \textbf{m}(\rho)= X^{-1}\left( \textbf{r}-\textbf{x} \right)\,.
    \ee
    Hence, we only need to minimise 
    \bb 
        V\longmapsto\Tr[h XVX^\intercal]
    \ee
    with respect to a covariance matrix $V$. By employing the Williamson's decomposition and the fact that the symplectic eigenvalues of a covariance matrix are always larger than one, it is simple to show that it suffices to minimise just over those covariance matrices of the form $V=SS^\intercal$, where $S$ is symplectic. Consequently, we need to evaluate the following minimisation problem:
    \bb 
        \min_{S\in\mathrm{Sp}(2n)}\Tr[ S^\intercal X^\intercal h XS ]\,,
    \ee
    where the minimisation is over the set $\mathrm{Sp}(2n)$ of all symplectic matrices. Note that since $h$ is strictly poisitive and since $X$ is invertible, it follows that $X^\intercal hX$ is strictly positive. Hence, we can apply the Williamson decomposition of $X^\intercal h X$ to show that
        \bb 
        \min_{S\in\mathrm{Sp}(2n)}\Tr[ S^\intercal X^\intercal h XS ]= \min_{S\in\mathrm{Sp}(2n)}\Tr[ S^\intercal D(X^\intercal h X)S ]\,,
    \ee
    where we denoted as $D(X^\intercal h X)$ the matrix of symplectic eigenvalues of $X^\intercal h X$. Finally, by employing~\cref{lem_min_s}, we conclude that
    \bb 
        \min_{S\in\mathrm{Sp}(2n)}\Tr[ S^\intercal D(X^\intercal h X)S ]=\Tr[ D(X^\intercal h X) ]=2\sum_{i=1}^nd_i(X^\intercal h X)\,,
    \ee
    where $d_1(X^\intercal h X),\ldots, d_n(X^\intercal h X)$ denote the symplectic eigenvalues of $X^\intercal h X$.
\end{proof}

\section{Conclusion}
\label{sec:conclusions}
In this work, we established several results on the thermodynamics of continuous variable systems. 
We gave a complete characterization of the Gaussian unitaries extracting the maximum energy from any given quantum state, yielding an analytical formula for Gaussian ergotropy in terms of the covariance matrix of the state and Hamiltonian.
Among the other results, we proved that Gaussian ergotropy is additive, i.e.~entangling Gaussian operations do not have any advantage over their local counterparts when many copies of the state are available.
Furthermore, we established new entropic bounds on the extractable energy under Gaussian operations and examined the problem of minimizing energy at the output of a Gaussian channel.

\begin{acknowledgments}
SFEO, VC and VG acknowledge financial support by MUR (Ministero dell’Università e della Ricerca) through the PNRR MUR project PE0000023-NQSTI.
\end{acknowledgments}

\textit{Note added:}
During the completion of this manuscript, we became aware of a closely related work~\cite{updike2025minimizingphasespaceenergies}, which studies and solves the problem of classical ergotropy for quadratic Hamiltonians, deriving \cref{lem_min_s} through an alternative approach.

\bibliographystyle{apsrev4-1}
\nocite{apsrev41Control}
\bibliography{biblio,revtex-custom}
\end{document}